\documentclass[12pt,a4paper]{article} 
\usepackage[margin=2cm]{geometry}
\usepackage[english]{babel}
\usepackage[latin1]{inputenc}
\usepackage{amsmath, amsthm, amssymb}
\usepackage[pdftex]{graphicx}
\usepackage{amssymb}
\usepackage{latexsym}
\usepackage{amstext}
\usepackage{epstopdf}
\newtheorem{theorem}{Theorem}[section]

\newtheorem{definition}[theorem]{Definition}

\newtheorem{proposition}[theorem]{Proposition}
\newtheorem{remark}[theorem]{Remark}

\newcommand{\ad}{\mathbf{a_2}}

\newcommand{\adc}{\mathbf{a_2^c}}

\newcommand{\adf}{\mathbf{a_2^f}}
\newcommand{\adgd}{\mathbf{a_2^{(2)}}}
\newcommand{\adgu}{\mathbf{a_2^{(1)}}}

\newcommand{\af}{\mathbf{a_5}}
\newcommand{\as}{\alpha}
\newcommand{\au}{a_1}

\newcommand{\bb}{\beta}
\newcommand{\bit}{\begin{itemize}}

\newcommand{\be}{\begin{equation}}
\newcommand{\bea}{\begin{equation}\begin{array}}
\newcommand{\beas}{\begin{equation*}\begin{array}}
\newcommand{\bef}{\begin{flalign}}
\newcommand{\befs}{\begin{flalign*}}
\newcommand{\bes}{\begin{equation*}}
\newcommand{\blms}{{\mathfrak B}_{LS}}
\newcommand{\cc}{\mathbf{C}}

\newcommand{\ddn}{\mathbf {d_N}}
\newcommand{\ddnd}{\mathbf {d_{N-2}}}
\newcommand{\ddnt}{\mathbf {d_{N-3}}}

\newcommand{\dq}{\mathbf{d_4}}
\newcommand{\dqn}{\mathbf{d_{4n}}}
\newcommand{\eab}{\varepsilon (\alpha,\beta)}

\newcommand{\ebg}{\varepsilon (\beta,\gamma)}
\newcommand{\ega}{\varepsilon (\gamma,\alpha)}
\newcommand{\egb}{\varepsilon (\gamma,\beta)}
\newcommand{\ee}{\end{equation}}
\newcommand{\eea}{\end{array}\end{equation}}
\newcommand{\eeas}{\end{array}\end{equation*}}
\newcommand{\eef}{\end{flalign}}
\newcommand{\eefs}{\end{flalign*}}
\newcommand{\ees}{\end{equation*}}

\newcommand{\eit}{\end{itemize}}

\newcommand{\eo}{\mathbf{e_8}}
\newcommand{\eon}{\mathbf{e_8^{(n)}}}
\newcommand{\ep}{\varepsilon}

\newcommand{\es}{\mathbf{e_6}}
\newcommand{\esn}{\mathbf{e_6^{(n)}}}

\newcommand{\est}{\mathbf{e_7}}
\newcommand{\estn}{\mathbf{e_7^{(n)}}}

\newcommand{\fff}{\mathfrak F}
\newcommand{\fq}{\mathbf{f_4}}
\newcommand{\fqn}{\mathbf{f_4^{(n)}}}

\newcommand{\gd}{\mathbf{g_2}}
\newcommand{\gdn}{\mathbf{g_2^{(n)}}}
\newcommand{\gh}{\gamma}

\newcommand{\gmf}{{\mathfrak g}}

\newcommand{\hri}{h_{\rho_i}}
\newcommand{\hrd}{h_{\rho_2}}
\newcommand{\hrt}{h_{\rho_3}}
\newcommand{\hru}{h_{\rho_1}}

\newcommand{\jdot}{\!\cdot\!}

\newcommand{\jotn}{\mathbf{J_3^{\pmb \nu}}}
\newcommand{\jobtn}{\mathbf{\overline J_3}^{\raisebox{-2 pt}{\scriptsize $\pmb \nu$}}}

\newcommand{\jotd}{\mathbf{J_3^2}}
\newcommand{\jobtd}{\mathbf{\overline J_3^{\raisebox{-2 pt}{\scriptsize \textbf 2}}}}

\newcommand{\joto}{\mathbf{J_3^8}}
\newcommand{\jobto}{\mathbf{\overline J_3^{\raisebox{-2 pt}{\scriptsize \textbf 8}}}}
\newcommand{\jp}{\circ}
\newcommand{\lk}{\mathfrak{L}}
\newcommand{\Ll}{\mathbb L}
\newcommand{\lms}{{\bf{\mathcal L_{MS}}}}
\newcommand{\lra}{\leftrightarrow}

\newcommand{\nbf}{{\pmb\nu}}
\newcommand{\nin}{\noindent}

\newcommand{\sref}[1]{{\bf\ref{#1}}}
\newcommand{\str}{\text{str}}

\newcommand{\tfo}{T_O}
\newcommand{\tfop}{T_O^\prime}
\newcommand{\tfs}{T_S}
\newcommand{\tfsm}{T_S^-}
\newcommand{\tfsp}{T_S^+}
\newcommand{\tfspm}{T_S^\pm}

\newcommand{\um}{{\scriptstyle \frac12}}

\newcommand{\xd}{x_{P2}}
\newcommand{\xdb}{\bar x_{P2}}
\newcommand{\xpi}{x_{Pi}}

\newcommand{\xpjb}{\bar x_{Pj}}
\newcommand{\xs}{x^\#}
\newcommand{\xt}{x_{P3}}
\newcommand{\xtb}{\bar x_{P3}}
\newcommand{\xu}{x_{P1}}
\newcommand{\xub}{\bar x_{P1}}

\newcommand{\zz}{\mathbb Z}

\numberwithin{equation}{section}
\begin{document}
%
\begin{titlepage}
\begin{center}
\hfill DFPD/2017/TH/14\\


{\huge{\bf The Magic Star of \\ Exceptional Periodicity\\}}
\vskip 0.5cm
\includegraphics[scale=0.52]{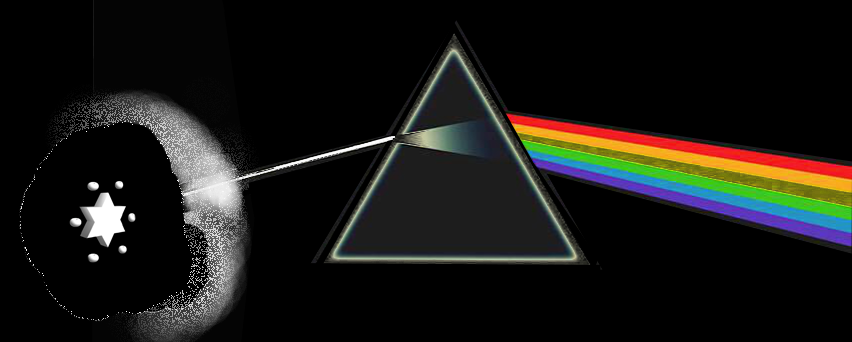}
\vskip 1.0cm

{\large{\bf Piero Truini\,$^1$, Michael Rios\,$^2$, Alessio Marrani\,$^{3,4}$}}

\vskip 40pt

{\it ${}^1$ Dipartimento di Fisica, Universit\` a di Genova,\\ and INFN, sez. di Genova,\\
via Dodecaneso 33, I-16146 Genova,  Italy} \vskip 5pt
\texttt{truini@ge.infn.it}

\vskip 10pt

{\it ${}^2$ Dyonica, ICMQG, USA} \vskip 5pt
\texttt{mrios@dyonicatech.com}

\vskip 10pt

{\em $^3$ Museo Storico della Fisica e Centro Studi e Ricerche ``Enrico Fermi'',\\
Via Panisperna 89A, I-00184, Roma, Italy \vskip 5pt }

\vskip 10pt

{\em $^4$ Dipartimento di Fisica e Astronomia ``Galileo Galilei'', Universit\`a di Padova,\\and INFN, Sez. di Padova,\\Via Marzolo 8, I-35131 Padova, Italy \vskip 5pt }

{\tt jazzphyzz@gmail.com} \\

    \vspace{10pt}



\end{center}

\vskip 0.6cm

\begin{center} {\bf ABSTRACT}\\[3ex]\end{center}
We present a periodic infinite chain of finite generalisations of the exceptional structures, including $\eo$, the exceptional Jordan algebra (and pair), and the octonions. We demonstrate that the exceptional Jordan algebra is part of an infinite family of finite-dimensional matrix algebras (corresponding to a particular class of cubic Vinberg's T-algebras). Correspondingly, we prove that $\eo$ is part of an infinite family of algebras (dubbed ``Magic Star'' algebras) that resemble lattice vertex algebras.

\vskip 1.0cm

\begin{center}
\textsl{\noindent {Seminar delivered by MR at the 4th Mile High Conference on Nonassociative Mathematics
University of Denver, Denver, Colorado, USA, July 29-August 5, 2017}}
\end{center}

\end{titlepage}

\newpage \setcounter{page}{1} \numberwithin{equation}{section}

\newpage\tableofcontents

\section{Introduction}
{\it In the mind of every mathematician, there is a tension between the general rule and exceptional
cases. Our conscience tells us we should strive for general theorems, yet we are fascinated and seduced by beautiful exceptions.}

J. Stillwell \cite{Stillwell}\\

Since the inception of quantum mechanics operator algebras and their
symmetries have proved to be of central importance. Jordan, Wigner and von
Neumann \cite{JWVN}, in a rigorous exploration of extended forms of quantum
mechanics, classified finite dimensional self-adjoint operator algebras, now
known as formally real Jordan algebras. The exceptional case they found
contained the largest division algebra, the octonions $\mathbb{O}$ \cite%
{Octonions}, first discovered by J. T. Graves in 1843.

With Gell-Mann's Eightfold way \cite{GM}, and subsequent study of the $su_{3}
$ quark and gluon structures \cite{quark-1, quark-2}, Lie algebras
solidified a firm role in the study of elementary particles. After the
classifications by Killing and Cartan, all finite-dimensional complex simple
Lie algebras, along with all their non-compact real forms, are known.
Besides the infinite classical series $\mathbf{a}_{n}$, $\mathbf{b}_{n}$, $\mathbf{c}_{n}$, $\mathbf{d}_{n}$,
five \textit{exceptional} Lie algebras exist : $\mathbf{g}_{2}$, $\mathbf{f}_{4}$, $\mathbf{e}_{6}$, $%
\mathbf{e}_{7}$ and $\mathbf{e}_{8}$. Gursey, Ramond and others argued exceptional Lie
algebras played a role in unification of fundamental particles \cite{unif}.
With such studies the octonions made a return, as the algebra of their
derivations is the smallest exceptional Lie algebra $\mathbf{g}_{2}$. Moreover the
maximal self-adjoint operator algebra over the octonions, the exceptional
Jordan algebra $\mathbf{J}_{\mathbf{3}}^{\mathbb{O}}$, has the next largest exceptional Lie algebra $\mathbf{f}_{4}$
describing its derivations \cite{f4}.

Various non-compact real forms of exceptional Lie algebras play the role of
electric-magnetic duality ($U$-duality\footnote{%
Here $U$-duality is referred to as the \textquotedblleft
continuous\textquotedblright\ symmetries of \cite{CJ-1}. Their discrete
versions are the $U$-duality non-perturbative string theory symmetries
introduced in \cite{HT-1}.}) algebras in locally supersymmetric theories of
gravity\footnote{%
Some non-compact real forms of exceptional algebras also occur in absence of
local supersymmetry (\textit{cfr. }\cite{Magic-Non-Susy}, and Refs. therein).%
} (see \textit{e.g.} \cite{MESGT} for theories with $8$ supersymmetries).

The largest exceptional Lie algebra $\mathbf{e}_{8}$ became central to the heterotic
string \cite{Het} construction, that assigned $16$ of the $26$ dimensions of
the bosonic string to the $\mathbf{e}_{8}\oplus \mathbf{e}_{8}$ even self-dual lattice.
Later, as $D=11$ supergravity merged with $D=10$ string theory in Witten's
\cite{9} mysterious $M$-theory, Ramond \textit{et al.} noticed the $D=11$
supermultiplet had a hidden $\mathbf{f}_{4}$ symmetry \cite{10}, an observation which
was further investigated by Sati \cite{Sati-1, Sati-2}. Ramond even
speculated whether the exceptional Jordan algebra $\mathbf{J}_{\mathbf{3}}^{\mathbb{O}}$ might be a special charge
space related to the $D=11$ lightcone \cite{11}, as it has natural $\mathbf{so}(9)$
and $\mathbf{f}_{4}$ symmetry\footnote{%
Recently, the maximally non-compact (\textit{i.e.}, split) real form $%
\mathbf{f}_{4(4)}$ has been conjectured as the global symmetry of an exotic
ten-dimensional theory in the context of the study of \textquotedblleft
Magic Pyramids" \cite{ICL-Magic}. See Mike Duff's talk at this Conference.}. Schwarz and Kim then recast the BFSS
matrix model \cite{19} for $M$-theory in terms of octonion variables \cite%
{12}, and Smolin invoked the full exceptional Jordan algebra in a
Chern-Simons string matrix model \cite{13} for Horowitz and Susskind's
conjectured \textquotedblleft bosonic $M$-theory" in $D=27$ \cite{14}.

With breakthroughs in algebraic geometry, coming from Connes and others in
the investigation of noncommutative geometry \cite{15,16,17}, a program for
emergent spacetime is underway and goes beyond Riemannian geometry in
favor of recovering manifolds from operator algebras. This move to an
algebraic derived spacetime proved to be novel, as it circumvented the usual
problems with Lorentz symmetry via discretization by making geometry
inherently fuzzy \cite{18}. It was even found that extended objects in
string theory, $D$-branes, had natural coordinates that were noncommutative
\cite{19,20}. This made the role of $C^{\ast }$-algebras and $K$-homology of
central importance to the study of $D$-branes \cite{21, 21-bis}, while also
providing the spectral triples for traditional noncommutative geometry and
its applications in deriving the Standard Model of particle physics \cite%
{22,23}.

In 2007, Lisi proposed a unified model \cite{L} using $\mathbf{e}_{8}$, which was
later discovered to be troubled \cite{DG10}. However, Lisi's theory inspired
Truini to study a special star-like projection - named \textit{%
\textquotedblleft Magic Star" }(MS) - of $\mathbf{e}_{8}$ under $\mathbf{a}_{2}$ \cite{Truini}%
; also motivated by the search of a unified way to characterize the fourth
row of the Freudenthal-Rozen-Tits Magic Square \cite{Magic Square}, the
Magic Star made the structure of Jordan algebras of degree three manifest in
each exceptional Lie algebra \cite{Truini} (see also \cite{Marrani-Truini-1,
Marrani-Truini-2}). An interaction based model using this $\mathbf{e}_{8}$ projection
has been speculated in \cite{Truini}, and proposed in \cite{Marrani-Truini-Interactions}.\medskip

Our aim, in this contribution to the Proceedings of 4th Mile High Conference on Nonassociative Mathematics, as well as in subsequent forthcoming works, is to define and study a consistent \textit{generalization} of exceptional Lie
algebras, by crucially exploiting some remarkable properties of the
MS.

We will show how it is possible to go beyond the largest finite-dimensional
exceptional Lie algebra $\mathbf{e}_{8}$ by relying on the basic structures
underlying the MS. This will result in the formulation of the so-called
\textit{"Exceptional Periodicity" }(EP), namely a generalization of
exceptional Lie algebras which is parametrized by a natural number $n\in
\mathbb{N}$, and which also enjoys a periodicity (ultimately traced back to
the well known Bott periodicity). At each \textquotedblleft level" of EP,
namely for each fixed value of $n$, the dimension of the resulting algebra
will be \textit{finite}, raising however the intriguing question (left for future studies) of investigating its $n\rightarrow \infty $ limit.
The impetus for the EP generalization came from certain $3$- and $5$-
gradings of the exceptional Lie algebras, along with spinor structures%
\footnote{%
Discussion with Eric Weinstein, during the \textquotedblleft Advances in
Quantum Gravity" symposium, San Francisco, July 2016.} beyond $\mathbf{e}_{8}$%
.

We will prove the existence of a \textit{periodic infinite chain of finite
generalisations of the exceptional structures}, including $\mathbf{e}_{8}$, the
Exceptional Jordan (or Albert) Algebra $\mathbf{J}_{\mathbf{3}}^{\mathbb{O}}$,
and the octonions $\mathbb{O}$ themselves. Remarkably, for $n=1$, the
MS-shaped structure of known finite-dimensional exceptional Lie algebras is
recovered.

As it will be evident from the subsequent treatment, the price to be paid
for such an elegant and periodic, $n$-parametrized and MS-shaped,
generalization of exceptional algebras is that \textit{the resulting
algebras will no longer be of Lie type}, namely they will not satisfy Jacobi
identities anymore. Indeed, within EP we will not be dealing with root
lattices, but rather with \textit{\textquotedblleft extended"} root
lattices, which will be thoroughly defined further below.

It is here worth remarking that EP provides a way to go beyond $\mathbf{e}_{8}$ which
is radically different from the way provided by affine and (extended)
Kac-Moody generalizations, such as\footnote{%
For recent development on $\mathbf{e}_{11}$ and beyond, see \cite{beyond e11}.} $%
\mathbf{e}_{8}^{+}=:\mathbf{e}_{9}$, $\mathbf{e}_{8}^{++}=:\mathbf{e}_{10}$, $e_{8}^{+++}=:\mathbf{e}_{11}$, which also
appeared as symmetries for (super)gravity models reduced to $D=2,1,0$
dimensions \cite{extended-Refs, West}, respectively, as well as near
spacelike singularities in supergravity \cite{spacelike-Refs}. In fact,
while such extensions of $\mathbf{e}_{8}$ are still of Lie type but \textit{%
infinite-dimensional}, the generalization of exceptional Lie algebras
provided by EP is \textit{not of Lie type}, but nevertheless is \textit{%
finite-dimensional}, for each level of the EP itself.

\bigskip

The paper is organized as follows.

In section \sref{sec:jp} we recall the notions of Jordan Algebras and Pairs with particular emphasis on their relationship with all the exceptional Lie algebras, except $\eo$. In section \sref{sec:magic} we introduce the Magic Star and show its main features, among which the appearance of a triple of Jordan Pairs at the core of $\eo$. This is, in our opinion,  the best way of expressing the link between $\eo$ and Jordan algebras.
In sections \sref{sec:ep}, \sref{sec:lms} and \sref{sec:ta} we extend the features of the Magic Star, in particular its relationship with the exceptional Lie algebras, to an infinite chain of finite dimensional algebras called {\it Exceptional Periodicity} (EP). In the case of the extension $\eon$ of $\eo$, the algebra has rank $N=4(n+1)$, $n=1,2,...$. In section \sref{sec:ep} we introduce a system of {\it extended} roots for EP, depending on $n$, which generalizes the root system of the exceptional Lie algebras. In section \sref{sec:lms} we associate an algebra, for each $n$, to the extended root system. It appears as a finite dimensional generalization of the exceptional Lie algebras, in which the Jacobi identity is only partially true. In section \sref{sec:ta} we show that in the Magic Star for EP the role of the Jordan algebras is played by matrix algebras, first introduced by Vinberg \cite{vin}, called T-algebras. All in all, EP not only generalizes exceptional Lie algebras, but also the exceptional Jordan algebra via its T-algebra \cite{vin} structure.  And, moreover, this structure is made manifest in the MS projections of the extended root systems. In section \sref{sec:fd} we list some future developments on the study of EP and its applications to a model for Quantum Gravity. Most of these developments are at an advanced stage at the time we are completing the present paper.\\

\section{Jordan Structures and related Lie algebras}\label{sec:jp}
{\it There are no Jordan algebras, there are only Lie algebras.}

I. L. Kantor\\

The modern formulation of Jordan algebras, \cite{jacob1}, involves a quadratic map $U_x y$ (like $xyx$ for associative algebras) instead of the original symmetric product $x \jdot y = \frac12(xy + yx)$, \cite{JWVN}. The quadratic map and its linearization $V_{x,y} z = (U_{x+z} - U_x - U_z)y$ (like $xyz+zyx$ in the associative case) reveal  the mathematical structure of Jordan Algebras much more clearly, through the notion of inverse, inner ideal, generic norm, \textit{etc}. The quadratic formulation is also particularly suited for the connection with Lie algebras, as we will explain in the next section. The axioms for (quadratic) Jordan algebras are:
\begin{equation}
U_1 = Id \quad , \qquad
U_x V_{y,x} = V_{x,y} U_x \quad  , \qquad
U_{U_x y} = U_x U_y U_x
\label{qja}
\end{equation}
The quadratic formulation led to the concept of Jordan Triple systems \cite{myb}, an example of which is a pair of modules represented by rectangular matrices. There is no way of multiplying two matrices $x$ and $y$ , say $n\times m$ and $m\times n$ respectively, by means of a bilinear product. But one can do it using a product like $xyx$, quadratic in $x$ and linear in $y$. Notice that, like in the case of rectangular matrices, there needs not be a unity in these structures. The axioms are in this case:
\begin{equation}
U_x V_{y,x} = V_{x,y} U_x \quad  , \qquad
V_{U_x y , y} = V_{x , U_y x} \quad , \qquad
U_{U_x y} = U_x U_y U_x
\label{jts}
\end{equation}

Finally, a Jordan Pair is defined just as a pair of modules $(V^+, V^-)$ acting on each other (but not on themselves) like a Jordan Triple:
\begin{equation}\begin{array}{ll}
U_{x^\sigma} V_{y^{-\sigma},x^\sigma} &= V_{x^\sigma,y^{-\sigma}} U_{x^\sigma}
\\
V_{U_{x^\sigma} y^{-\sigma} , y^{-\sigma}} &= V_{x ^\sigma, U_{y^{-\sigma}} x^\sigma} \\
U_{U_{x^\sigma} y^{-\sigma}} &= U_{x^\sigma} U_{y^{-\sigma}} U_{x^\sigma}\end{array}
\label{jp}
\end{equation}
where $\sigma = \pm$ and $x^\sigma \in V^{+\sigma} \, , \; y^{-\sigma} \in V^{-\sigma}$.

Jordan Pairs are strongly related to the Tits-Kantor-Koecher construction of Lie Algebras $\lk$ \cite{tits1}-\nocite{kantor1}\cite{koecher1} (see also the interesting relation to Hopf algebras, \cite{faulk}):
\begin{equation}
\lk = J \oplus \str(J) \oplus \bar{J} \label{tkk}
\end{equation}
where $J$ is a Jordan algebra and $\str(J)= L(J) \oplus Der(J)$ is the structure algebra of $J$ \cite{mac1}; $L(x)$ is the left multiplication in $J$: $L(x) y = x \jdot y$ and $Der(J) = [L(J), L(J)]$ is the algebra of derivations of $J$ (the algebra of the automorphism group of $J$) \cite{schafer1}\cite{schafer2}.

 In the case of complex exceptional Lie algebras, this construction applies to $\est$, with $J = \joto\equiv \mathbf{J}_{\mathbf{3}}^{\mathbb{O}}$, the 27-dimensional exceptional Jordan algebra of $3 \times 3$ Hermitian matrices over the complex octonions, and $\str(J) = \es \otimes \mathbb{C}$ - $\mathbb{C}$ denoting the complex field. The algebra $\es$ is called the \emph{reduced structure algebra} of $J$, $\str_0(J)$, namely the structure algebra with the generator corresponding to the multiplication by a complex number taken away: $\es = L(J_0) \oplus Der(J)$, with $J_0$ denoting the traceless elements of $J$.

Let us focus on exceptional Jordan and Lie algebras. We consider the Jordan algebras and Pairs based on the $3\times 3$ matrices $\jotn, \pmb\nu= 1,2,4,8$:
\be
\jotn := \left( \begin{array}{ccc} \alpha & a & \bar b \\ \bar a & \beta & c \\ b & \bar c & \gamma \end{array} \right) \
\begin{array}{l}  \alpha \, , \beta \, , \gamma \in \mathbb{C} \; ; \, a\, , b \, , c \in \mathbb{A}\\  {\pmb\nu}=1,2,4,8 \text{ for } \mathbb{A} = \mathbb{R} , \, \mathbb{C}  ,\, \mathbb{H}  ,\, \mathbb{O} \end{array}
\label{jm}
\ee
where $\mathbb{R} , \, \mathbb{C}  ,\, \mathbb{H}  ,\, \mathbb{O}$ denote the Hurwitz algebras of reals, complex, quaternions and octonions, respectively.\\
If $x,y \in \jotn$ and $xy$ denotes their standard matrix product, we denote by $x\jdot y := \frac12 (xy + yx)$ the Jordan product of $x$ and $y$. The Jordan identity is the power associativity with respect to this product:
\be\label{pass}
x^2 \jdot (x\jdot z) - x \jdot (x^2 \jdot z) = 0,
\ee

Another  fundamental product is the {\it sharp} product $\#$, \cite{mac1}. It is the linearization of $\xs := x^2 - t(x) x - \frac12(t(x^2) - t(x)^2)I$, with $t(x)$ denoting the trace of $x\in \jotn$, in terms of which we may write the fundamental cubic identity for $\jotn, \pmb\nu= 1,2,4,8$:
\be\label{cubic}
\xs\jdot  x = \frac13 t(\xs\!, x) I \quad \text{or} \quad x^3 - t(x) x^2 + t(\xs) x - \frac13 t(\xs\! , x) I = 0
\ee
where we use the notation $t(x,y) := t(x\jdot y)$ and  $x^3 = x^2 \jdot x$ (notice that for $\joto$, because of non-associativity, $x^2 x \ne x x^2$ in general).

The triple product is defined as, \cite{mac1}:
\bea{ll}\label{vid}
\{ x , y , z \} := V_{x,y}z :&= t(x,y) z + t(z,y) x - (x \# z) \# y \\
&= 2 \left[ (x \jdot y)\jdot z +  (y \jdot z)\jdot x - (z \jdot x)\jdot y \right]
\eea

Notice that the last equality of \eqref{vid} is not trivial at all. $V_{x,y}z$ is the linearization of the quadratic map $U_xy$. The equation (2.3.15) at page 484 of \cite{mac1} shows that:
\be\label{uid}
U_x y = t(x,y) x  - x ^\# \# y = 2 (x \jdot y)\jdot x - x^2 \jdot y
\ee

The following identities can be derived from the Jordan Pair axioms, \cite{loos1}:
\be
\left[ V_{x,y} , V_{z,w} \right] = V_{{V_{x,y} z},w} - V_{z,{V_{x,y} w}}
\label{comv}
\ee
and, for $D = (D_+,D_-)$ a derivation of the Jordan Pair $V$ and $\beta(x,y) = (V_{x,y}, - V_{y,x})$,
\be
[D, \beta(x,y)] = \beta (D_+(x),y) + \beta(x, D_-(y))
\label{dib}
\ee

We have:
\nin For $J_P:=(J,\bar J)$: $str(J) = Der(J_P)$
$$U_{x^+}y^- = [[x^+,y^-],x^+]\quad ; \quad U_{x^-}y^+ = [[x^-,y^+],x^-]$$
JP axioms follow from the Jacobi identity of the Lie algebra $str(J)$.\\

To summarize the relationship between the exceptional Jordan algebra and the exceptional Lie algebras is:\\

\nin $\gd = Der(\mathbb{O})$\\

\nin$\fq = Der(\joto)$ ; $Der(J) = [L(J), L(J)]$\\

\nin$\es =str_o(\joto) = L(J_0) \oplus Der(J) =Der_o(\joto,\jobto)$\\

\nin$\est = \joto \oplus Der(\joto,\jobto) \oplus \jobto$ \ \text{superstructure algebra}\\

\nin $\eo = $ {\bf ?}

The link between $\eo$ and Jordan Algebras and Pairs is expressed, in our opinion, by the Magic Star, \cite{Truini}.


\section{The Magic Star: Exceptional Lie algebras and Jordan structures}\label{sec:magic}
{\it The star that can bring peace and order back to the chaotic world.}

citation from the ``Magic Star'' TV series, China, 2017.\\

\begin{figure}[h]\centering
\includegraphics[scale=1]{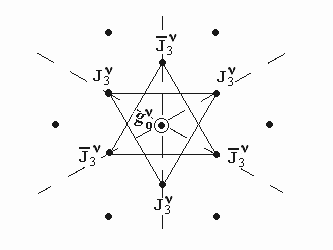}
\caption{Root diagram of $\fq$, $\es$, $\est$, $\eo$ (${\pmb \nu}=1,2,4,8$) projected on the plane of one $\ad$}\label{fig:Fig1}
\end{figure}

The Magic Star, \cite{Truini} \cite{Marrani-Truini-1}, shown in Figure \ref{fig:Fig1}, is the projection of the roots of the
exceptional Lie algebras on a complex $\mathbf{su(3)}=\ad$ plane,
recognizable by the dots forming the external hexagon, and it exhibits the
\textit{Jordan Pair} content of each exceptional Lie algebra. There are
three Jordan Pairs $(\jotn,\jobtn)$, each of which lies on an axis
symmetrically with respect to the center of the diagram. Each pair doubles
a simple Jordan algebra of rank $3$, $\jotn$, with involution - the
conjugate representation $\jobtn$, which is the algebra of $3\times 3$
Hermitian matrices over $\mathbb{A}$, where $\mathbb{A}=\mathbb{R},\,\mathbb{C}  ,\,\mathbb{H} ,\,\mathbb{O}$ for $\pmb\nu=1,2,4,8$ respectively.
The exceptional Lie algebras $\fq$, $\es$, $\est$, $\eo$ are obtained for $\nbf%
=1,2,4,8$, respectively. $\gd$ can be also represented in the same way, with
the Jordan algebra reduced to a single element. The Jordan algebras $\jotn$ (and their conjugate $\jobtn$) globally
behave like a $\mathbf{3}$ (and a $\mathbf{\overline{3}}$) dimensional
representation of the outer $\mathbf{a}_{2}$. The algebra denoted by $\mathbf{%
g_{0}^\nbf}$ in the center (plus the Cartan generator associated with the
axis along which the pair lies) is the algebra of the
automorphism group of the Jordan Pair; namely, $\mathbf{g_{0}^\nbf}$ is the
the \textit{reduced} structure group of the corresponding
Jordan algebra $\jotn$: $\mathbf{g_{0}^\nbf=str}%
_0\left( \jotn \right)$.

The \textit{base field} considered throughout the present paper is $\mathbb{C}$.
The various real compact and non-compact forms
of the exceptional Lie algebras follow as a consequence, using some more or
less laborious tricks, whose treatment we leave to a future study, and they
do not affect the essential structure.\\

What is magic in the Magic Star? Let us list a number of features which are clearly shown by the picture.\\

\bit
\item[1)] The Magic Star represents a unifying view of all the exceptional algebras.

\item[2)] Since $Der_o(\jotd,\jobtd )=\ad\oplus \ad$ and
$Der_o(\joto,\jobto )=\es$, Figure \ref{fig:Fig1} depicts the following decomposition, \cite{Truini} :%
\be\begin{array}{ll}\label{eodec}
\eo &= \adc+(3,\joto) + (\bar 3, \jobto) + Der_o(\joto,\jobto)\\
&= \adc+(3,\joto) + (\bar 3, \jobto) + \adf+(3,\jotd) + (\bar 3, \jobtd) + Der_o(\jotd,\jobtd)\\
&= \adc+(3,\joto) + (\bar 3, \jobto) + \adf+(3,\jotd) + (\bar 3, \jobtd) + \adgu + \adgd
\end{array}\ee
which shows that $\eo$ is made of 4 orthogonal $\ad$ subalgebras, plus 3 Jordan Pairs $(\joto,\jobto)$ over the octonions, plus 3 Jordan Pairs $(\jotd,\jobtd)$ over the complex. The role of the $\ad$'s is interchangeable: any $\ad$ would produce a Magic Star. Select one and call it $\adc$ for {\it color $\ad$} and another one, $\adf$, inside $\es$ as  {\it flavor $\ad$}, then the Jordan Pairs in \ref{eodec} represent three families of colored and three families of flavor degrees of freedom (notice that the flavor part is a singlet of $\adc$).The two remaining $\ad$'s may be related to gravity. We are currently working on this aspect with the hope to produce an $\eo$ model of quantum gravity with an emergent discrete spacetime.

\item[3)] By considering a Jordan Pair $(\joto,\jobto)$ plus $\es\oplus \mathbb{C}$ we get the 3-graded Lie algebra $\est$, clearly shown in Figure \ref{fig:est}.
\begin{figure}[h!]\centering
\includegraphics[scale=1]{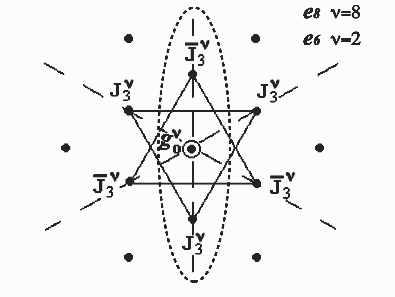}
\caption{$\est$ (resp. $\af$) as 3-graded inside $\eo$ (resp. $\es$)}\label{fig:est}
\end{figure}
This shows why quadratic maps and Jordan Pairs are the natural Jordan product and structure inside $\eo$: $\ U_{x^\sigma}y^{-\sigma} = [[x^\sigma,y^{-\sigma}],x^\sigma] \ , \ \sigma=\pm$\\
whereas $[x^\sigma,y^{\sigma}]=0$ as clearly shown by the 3-grading.

\item[4)] Other important algebraic structures easily visible in the Magic Star are related to the 5-grading of $\eo$: the Freudenthal Triple Systems (FTS), \cite{hel}, in the case of $\eo$ an irreducible 56-dimensional irreducible representation of $\est$, and the Kantor Pair, \cite{kp}, formed by two FTS'. They are shown in Figure \ref{fig:ftskp}.
\begin{figure}[htbp]
\includegraphics[width=150mm,keepaspectratio]{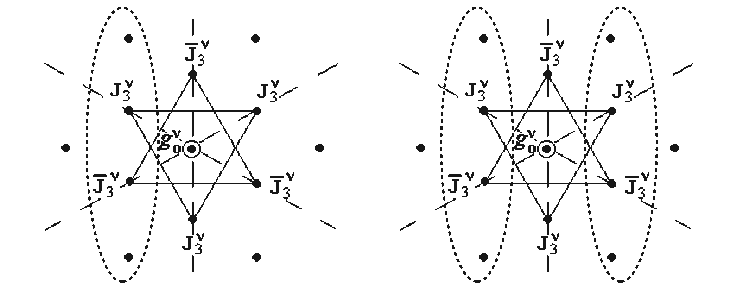}
\caption{FTS (left) and KP (right) within $\eo$}\label{fig:ftskp}
\end{figure}
Denoting $\gmf = \gmf_{-2} + \gmf_{-1} + \gmf_0 + \gmf_1 + \gmf_2$ the 5-grading, we have $\est$ in grade 0  = $Der(FTS)$, an FTS in grade 1 (a 56-dim irrep of $\est$) and 1-dimensional $\gmf_{\pm 2} := \cc x_{\pm \rho}$. The non-deg skew bilinear $<\ , \ >$ and a fully symmetric quartic form $q$ on $\gmf_1$, characteristic of the FTS, are obtained by:
\bea{rcll}
[ x, y] &=& <x,y> x_\rho & x,y\in \gmf_1\\
q(x,x,x,x) &=& [x , [x , [x , [x , x_{-\rho}]]]] &  x\in \gmf_1
\eea
from which $\Rightarrow q(x_{\beta_1}, x_{\beta_2} , x_{\beta_3} , x_{\beta_4})$ follows by linearization.\\

The Kantor Pair is a pair of modules and within $\eo$ is $$KP = \eo \ominus (\est\oplus \mathbf{a}_{1})$$
\eit


\section{Exceptional Periodicity and Extended Roots}\label{sec:ep}
{\it Divinity drew Earth out of Emptiness as he drew One from Nothing to create many.}

Pythagoras\\

We force the definition of root system to include what we call {\it extended} roots not obeying the symmetry by Weyl reflection, nor the fact that $2\dfrac{(\alpha,\beta)}{(\alpha,\alpha)}$ be integer for all roots $\alpha$, $\beta$.\\

Let $\{k_1 , ... , k_N\}$ be an orthonormal basis of an Euclidean space $V$ of dimension $N$. For any $n=1,2,...$ we introduce $N=4(n+1)$ and define the extended roots of $\eon$ as:

\nin
{\Large $\eon$}:
\bea{lllcl}
\pm k_i\pm k_j & 1\le i<j\le N && 2N(N-1) &\text{roots}\\
\frac12 (\pm k_1 \pm k_2 \pm ... \pm k_N) &\text{even \# of +} && 2^{N-1} &\text{roots}
\eea

This is a root system only in the case $n=1$ and $\eo^{(1)}=\eo$.\\

The {\it extended} roots of $\gdn$ are obtained by projecting on the space spanned by $k_1-k_2$ and $k_1+k_2-2k_3$, hence $\gdn=\gd$; those of $\fqn$ are the projection on the space  spanned by $k_1,k_2,...,k_{N-4}$:

\nin
{\Large $\fqn$}:
\bea{ll}
\pm k_i \ , \ \pm k_i\pm k_j & 1\le i<j\le N-4\\
\frac12 (\pm k_1 \pm k_2 \pm ... \pm k_{N-4}) &
\eea

and finally:\\

\nin
{\Large $\esn$}:
\bea{ll}
\pm k_i\pm k_j & 1\le i<j\le N-3\\
\frac12 (\pm k_1 \pm k_2 \pm ... \pm (k_{N-2}+k_{N-1}+k_N)) &\text{even \# of +}
\eea

\nin
{\Large $\estn$}:
\bea{ll}
\pm(k_{N-1}+k_N)\\
\pm k_i\pm k_j & 1\le i<j\le N-2\\
\frac12 (\pm k_1 \pm k_2 \pm ... \pm (k_{N-1}+k_N)) &\text{even \# of +}
\eea

\begin{figure}[htbp]\centering
\includegraphics[width=80mm]{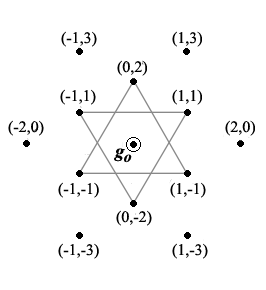}
	\caption{The Magic Star}\label{fig:ms}
\end{figure}

All these sets of {\it extended} roots form a Magic star of figure \ref{fig:ms}, once projected on the plane spanned by $k_1-k_2$ and $k_1+k_2-2k_3$. In the figure the pair of integers $(a,b)$ are the (Euclidean) scalar products of each root with $k_1-k_2$ and $k_1+k_2-2k_3$.\\

One can readily check that, upon a relabeling of the $k$'s, $\esn$ is the center of the Magic Star of $\eon$ and that $\estn=\esn\oplus T_{(a,b)}\oplus T_{(-a,-b)}$, for a fixed pair $(a,b)\in \{(1,1),(-1,1),(0,-2)\}$, where $\esn$ is the center of the Magic Star and $T_{(a,b)}$ is the set of roots of $\eon$ whose scalar products with $k_1-k_2$ and $k_1+k_2-2k_3$ are $(a,b)$.\\

The {\it rank} of $\gdn$, $\fqn$, $\esn$, $\estn$, $\eon$ is the dimension of the vector space $V$ spanned by their roots, namely $2$, $N-4$, $N-2$,  $N-1$, $N$ respectively.\\

By abuse of definition we shall often say {\it root}, for short, instead of
{\it extended root}.\\

From now on we restrict to $\esn$, $\estn$, $\eon$ and denote by $\lms$ anyone of them, by $\Phi$ the set of extended roots of $\lms$ and by $R$ its rank. We recall that $N=4(n+1)$, $n=1,2,...$, hence $R=N-2=4n+2$ for $\esn$, $R=N-1=4n+3$ for $\estn$, $R=N=4n+4$ for $\eon$.\\

We denote by $\Phi_O$ and $\Phi_S$ the following subsets of $\Phi$:
\bea{l}
\Phi_O = \{ (\pm k_i \pm  k_j) \in \Phi \}
\\ \\
\Phi_S =  \{ \frac{1}{2} (\pm k_1 \pm k_2 \pm ... \pm k_N)  \in \Phi \}
\eea

\begin{remark}\label{r:dd} Notice that $\Phi_O$ is the root system of $\ddnt$ in the case of $\esn$, of $\ddnd\oplus \au$ in the case of $\estn$ and of $\ddn$ in the case of $\eon$.
\end{remark}

We now prove the following results.

\begin{proposition} For all $\rho\in \Phi_O$ and $x\in \Phi$:
 $2\dfrac{(x, \rho)}{(\rho,\rho)}\in \zz$ and $w_\rho(x) = x - 2\dfrac{(x, \rho)}{(\rho,\rho)}\rho\in \Phi$ (the set of extended roots is closed under the Weyl reflections by all $\rho\in \Phi_O$).  The set of extended roots is closed under the Weyl reflections by all $\rho\in \Phi$ if and only if $n=1$.
\end{proposition}
\begin{proof} If $\rho\in \Phi_O$ then $(\rho,\rho)=2$ and $(x,\rho)\in \{0,\pm1,\pm 2\}$, hence $2\dfrac{(x, \rho)}{(\rho,\rho)}\in \zz$ and $w_\rho(x) = x - (x, \rho) \rho$. If $(x,\rho)=0$ then $w_\rho(x) = x\in \Phi$. If $(x,\rho)=\pm 1$ then $w_\rho(x) = x\mp \rho \in \Phi$ as we shall prove in proposition \ref{sproots}. If $(x,\rho)=\pm 2$ then $\rho=\pm x$ and $w_\rho(x) = - x\in \Phi$. Suppose now that both $x,\rho\in \Phi_S$ and write $x= \frac12 \sum \lambda_i k_i$, $\rho= \frac12 \sum \mu_i k_i$ where $\lambda_i,\mu_i\in \{-1,1\}$. We can certainly pick an $x$ such that $(x,\rho)=-n$ which occurs whenever for two  indices $j,\ell$ $\lambda_j=\mu_j$ and $\lambda_\ell=\mu_\ell$ while $\lambda_i = - \mu_i$ for $i\ne j,\ell$. Hence $w_\rho(x) = x + 2\dfrac{n}{n+1}\rho = \frac12\sum \nu_i k_i$. We have that $|\nu_j|= \left| \lambda_j + 2\dfrac{n}{n+1} \mu_j \right|= \dfrac{3n+1}{n+1}\ge 2$ and $|\nu_j| = 2$ if and only if $n=1$, in which case $\Phi$ is the root system of a simple Lie algebra. For $n>1$ there is no root with such a coefficient $\nu_j$.\hfill
\end{proof}

We introduce the basis $\Delta = \{\alpha_1 , ... ,\alpha_R\}$ of $\Phi$, with $\alpha_i = k_i-k_{i+1}\, , \ 1\le i \le R-2$, $\alpha_{R-1} = k_{R-2}+k_{R-1}$ and $\alpha_R = -\frac12(k_1+k_2+...+k_N)$; we order them by setting $\alpha_i > \alpha_{i+1}$:
\be\label{sroots}
\Delta = \{k_1-k_2, k_2-k_3, ..., k_{R-2} - k_{R-1}, k_{R-2} + k_{R-1}, -\um(k_1+k_2+...+k_N)\}
\ee

\begin{proposition}\label{approp} The set $\Delta$ in \eqref{sroots} is a set of {\bf simple extended roots}, by which we mean:

\bit
\item[i)] $\Delta$ is a basis of the Euclidean space $V$ of finite dimension $R$;
\item[ii)] every root $\beta$ can be written as a linear combination of roots of $\Delta$ with all positive or all negative integer coefficients: $\beta = \sum \ell_i \alpha_i$ with $\ell_i \ge 0$ or $\ell_i\le 0$ for all $i$.
\eit
For all roots $\beta = \sum_{i = 1}^R {m_i\alpha_i}$ the coefficient $m_R$ is such that:
\bea{ll}
m_R \in \{0,\pm 2\} & \text{if } \beta \in \Phi_O\\
m_R =\pm 1 & \text{if } \beta \in \Phi_S\\ \label{rema}
\eea
\end{proposition}

\begin{proof} The set $\Delta = \{\alpha_1 , ... ,\alpha_R\}$ is obviously a basis in $V$.
Let ${\bf u} = k_{N-2}+k_{N-1}+k_N, k_{N-1}+k_N, k_N$ for $\esn$, $\estn$, $\eon$ respectively. We have:
\bea{ccl} \label{kal}
k_{R-1} &= &\frac12 (\alpha_{R-1} - \alpha_{R-2})\\ \\
k_i &= &\alpha_i + k_{i+1} = \sum_{\ell=i}^{R-2}\alpha_\ell +\frac12 (\alpha_{R-1}-\alpha_{R-2}) \ , \ 1\le i\le R-2\\ \\
{\bf u} &= &- 2 \alpha_R - \sum_{\ell=1}^{R-2} {\ell \alpha_\ell - \frac{R-1}{2}(\alpha_{R-1}}-\alpha_{R-2})
\eea
from which we obtain, for $1\le i<j\le R-1$ and forcing $\sum_{\ell=r}^s{\alpha_\ell}=0$ if $r>s$:

\be\label{posroots1}
\left.
\begin{array}{rcl}
k_i + k_j &= & \sum_{\ell=i}^{R-3}\alpha_\ell + \sum_{\ell=j}^{R-2}\alpha_\ell + \alpha_{R-1}\\ \\
k_i - k_j &= & \sum_{\ell=i}^{j-1}\alpha_\ell
\end{array}
\right\}
1\le i<j\le R-1
\ee
for $\estn$, namely for $R=N-1=4n+3$ and ${\bf u}=k_{N-1}+k_N$:
\bea{rcl}\label{posroots2}
-{\bf u} &= & 2 \alpha_R + \sum_{\ell=1}^{R-3} {\ell \alpha_\ell} + 2n\, \alpha_{R-2} + (2n+1)\alpha_{R-1} \ , \text{ for } \estn
\eea
for $\eon$, namely for $R=N$:
\bea{lcl}
\pm k_i - k_N &= & 2 \alpha_N +\sum_{\ell=1}^{i-1} {\ell \alpha_\ell}+ \sum_{\ell=i}^{N-3} {(\ell\pm1)\alpha_\ell} + (2n+ \frac{1 \pm 1}{2})\alpha_{N-2}\\ \\
&&+ (2n+1 + \frac{1 \pm 1}{2})\alpha_{N-1} \ , \quad i\le N-2\\ \\
\pm k_{N-1} - k_N &= & 2 \alpha_N + \sum_{\ell=1}^{N-3} {\ell \alpha_\ell} + (2n+\frac{1 \mp 1}{2})\alpha_{N-2}+ (2n+1 + \frac{1 \pm 1}{2})\alpha_{N-1}
\label{posroots3}
\eea

We see that all the roots in \eqref{posroots1},\eqref{posroots2},\eqref{posroots3} are the sum of simple roots with all positive integer coefficients.
These are half of the roots in $\Phi_O$ and they are all positive roots. \\
The rest of the roots in $\Phi_O$ are negative and are obviously the sum of simple roots with integer coefficients that are all negative.

Finally all the roots in $\Phi_S$ that contain $-\frac12 {\bf u}$ can be obtained from $\alpha_R$ by flipping an even number of signs and this is done by adding to $\alpha_R$ a certain number of terms of the type $k_i+k_j$, $1\le i<j\le R-1$. These are all positive roots, a half of the roots in $\Phi_S$ and are linear combination of simple roots with integer coefficients that are all positive.\\
The negative roots are similarly obtained by adding to $-\alpha_N$ a certain numeber of terms of the type $-(k_i+k_j),\, 1\le i<j\le R-1$, and are linear combination of simple roots with integer coefficients that are all negative.\\
As an consequence we easily obtain \ref{rema}. \hfill
\end{proof}

\begin{proposition} \label{sproots}
For each $\alpha \in \Phi_O , \beta \in \Phi$ the scalar product $(\alpha,\beta) \in \{\pm 2, \pm 1, 0\}$; $\alpha + \beta$ ( respectively $\alpha - \beta$) is a root if and only if $(\alpha , \beta) = -1$ (respectively $+1$); if both $\alpha+\beta$ and  $\alpha-\beta$ are not in $\Phi\cup\{0\}$ then  $(\alpha,\beta)=0$.\\
 For each $\alpha , \beta \in \Phi_S$ the scalar product $(\alpha,\beta) \in \{\pm (n+1), \pm n , \pm (n-1), ..., 0 \}$; $\alpha + \beta$ ( respectively $\alpha - \beta$) is a root if and only if $(\alpha , \beta) = -n$ (respectively $+n$).\\
For $\alpha , \beta \in \Phi$ if $\alpha+\beta$ is a root then $\alpha - \beta$ is not a root.
\end{proposition}
\begin{proof} If both $\alpha , \beta \in \Phi_O$ the proof follows from the fact that $\Phi_O$ is the root system of a simply laced Lie algebra. If $\alpha \in \Phi_O , \beta \in \Phi_S$ then obviously $(\alpha,\beta) \in \{\pm 1, 0\}$. Moreover, let us write $\alpha = \sigma_i k_i +\sigma_j k_j \, , \ i < j$, $\sigma_{i,j}\in \{-1,1\}$. Then $(\alpha , \beta) = -1$ if and only if $\beta = \frac12 (\pm k_1 \pm ... -\sigma_i k_i \pm ... -\sigma_j k_j \pm ... \pm k_N)$, which is true if and only if $\alpha + \beta \in \Phi$ (in particular $\alpha + \beta \in \Phi_S$). Similarly $(\alpha , \beta) = 1$ if and only if $\alpha - \beta \in \Phi$. As a consequence, if both $\alpha\pm\beta \notin \Phi\cup\{0\}$ then  $(\alpha,\beta)\ne \pm 1$ and also $(\alpha,\beta)\ne \pm 2$ because $(\alpha,\beta)= \pm 2$ if and only if $\alpha = \pm \beta$; therefore $(\alpha,\beta)=0$.\\
If both $\alpha , \beta \in \Phi_S$, then all their signs but an even number $2m$ must be equal, $m=0,...,N/2=2(n+1)$ and we get $(\alpha,\beta) = \frac14 (N-2m - 2m) = n+1-m = n+1, n, n-1, ... ,-(n+1)$ for  $m=0,...,2(n+1)$. Moreover, since $\pm k_i \pm k_i  \in \{0,\pm 2 k_i\} \, , \ i = 1, ... N$ then $\alpha + \beta \in \Phi$ if and only if all signs are opposite but 2 (in which case $\alpha + \beta$ is actually in $\Phi_O$) and this is true if and only if $(\alpha , \beta) = -\frac14(N-4)=-n$. Similarly $(\alpha , \beta) = n$ if and only if $\alpha - \beta \in \Phi$. The last statement of the Proposition follows trivially. \hfill
\end{proof}

The preceding propositions extend to EP analogous results for the root system of simply laced Lie algebras. They are essential in determining the structure of the algebra $\lms$ that we associate to the extended roots in the next section.

\section{The $\lms$ algebra}\label{sec:lms}

{\it Mathematics is the art of giving the same name to different things.}

H. Poincar\'{e}\\

We define the $\lms$ algebra (as before $\lms$ is either $\esn$ or $\estn$ or $\eon$) by extending the construction of a Lie algebra from a root system, \cite{carter} \cite{hum} \cite{graaf}. In particular we generalize the algorithm in \cite{graaf} for simply laced Lee algebras, since also in our set of extended roots the $\beta$ chain through $\alpha$, namely the set of roots $\alpha+c\beta$, $c\in \zz$, has length one.\\

We give $\lms$ an algebra structure of rank $R$ over a field extension $\fff$ of the rational integers $\zz$ in the following way
\footnote{Specifically, we will take $\fff$ to be the complex field $\mathbb{C}$.}
:
\bit
\item[a)]  we select the set of simple extended roots $\Delta = \{\alpha_1 , ... ,\alpha_R\}$ of $\Phi$
\item[b)] we select a basis $\{ h_1 ,...,h_R\}$ of the $R$-dimensional vector space $H$ over $\fff$ and set $h_\alpha = \sum_{i=1}^R c_i h_i$ for each $\alpha  \in \Phi$ such that $\alpha = \sum_{i=1}^R c_i \alpha_i$
\item[c)] we associate to each $\alpha  \in \Phi$ a one-dimensional vector space $L_\alpha$ over $\fff$ spanned by $x_\alpha$
\item[d)] we define $\lms = H \bigoplus_{\alpha \in \Phi} {L_\alpha}$ as a vector space over $\fff$
\item[e)] we give $\lms$ an algebraic structure by defining the following multiplication on the basis
$\blms = \{ h_1 ,...,h_R\} \cup \{x_\alpha \ | \ \alpha \in \Phi\}$, extended by linearity to a bilinear multiplication $\lms\times \lms\to \lms$:
	\be\begin{array}{ll}
	&[h_i,h_j] = 0 \ , \ 1\le i, j  \le R \\
	&[h_i , x_\alpha] = - [x_\alpha , h_i] = (\alpha, \alpha_i )\, x_\alpha \ , \ 1\le i \le R \ , \ \alpha \in \Phi \\
	&[x_\alpha, x_{-\alpha} ] = - h_\alpha\\
	&[x_\alpha,x_\beta] = 0 \ \text{for } \alpha, \beta \in \Phi \ \text{such that } \alpha + \beta \notin			 	\Phi \ \text{and } \alpha \ne - \beta\\
	&[x_\alpha,x_\beta] = \varepsilon (\alpha , \beta)\,  x_{\alpha+\beta}\ \text{for } \alpha , \beta \in \Phi \ \text{such that }  \alpha+ \beta \in \Phi\\
	\end{array} \label{comrel}\ee
\eit

where $\varepsilon (\alpha , \beta)$ is the {\it asymmetry function}, introduced in \cite{kac}, see also \cite{graaf}, defined as follows:\\
\begin{definition} Let $\Ll$ denote the lattice of all linear combinations of the simple extended roots with integer coefficients
\be
\Ll = \left\{ \sum_{i=1}^R c_i \alpha_i \ |\ c_i \in \zz \ , \ \alpha_i \in \Delta \right\}
\label{lattice}
\ee
the asymmetry function $\varepsilon (\alpha , \beta) : \ \Ll \times \Ll \to \{-1,1\}$ is defined by:
\be\label{epsdef}
\varepsilon (\alpha , \beta) = \prod_{i,j=1}^R \varepsilon (\alpha_i , \alpha_j)^{\ell_i m_j} \quad \text{for } \alpha = \sum_{i=1}^R \ell_i\alpha_i \ ,\ \beta = \sum_{j=1}^R m_j \alpha_j
\ee
where $\alpha_i , \alpha_j \in \Delta$ and
\be
\varepsilon (\alpha_i , \alpha_j) = \left\{
\begin{array}{ll}
-1 & \text{if } i=j\\ \\
-1 & \text{if } \alpha_i + \alpha_j  \text{ is a root and } \alpha_i < \alpha_j\\ \\
+ 1 & \text{otherwise}
\end{array}
\right.
\ee
\end{definition}

We now show some properties of the asymmetry function. In particular we show that, for $\as,\bb,\as+\bb\in\Phi$, $\varepsilon (\alpha , \beta) = - \varepsilon (\beta , \alpha)$ which implies that the bilinear product \eqref{comrel} is antisymmetric.

\begin{proposition} \label{epsprop}The asymmetry function $\varepsilon$ satisfies, for $\alpha , \beta, \gamma , \delta \in \Ll$, $\alpha = \sum{m_i \alpha_i}$ and $\beta = \sum{n_i \alpha_i}$:
\bes\begin{array}{rrcl}
i) & \varepsilon (\alpha + \beta, \gamma) & =& \varepsilon (\alpha , \gamma)\varepsilon (\beta , \gamma) \\
ii) &\varepsilon (\alpha ,\gamma + \delta) & = &\varepsilon (\alpha , \gamma)\varepsilon (\alpha , \delta) \\
iii) &\varepsilon (\alpha , \alpha) & = &(-1)^{\frac12 (\alpha ,\alpha)-m_R^2 \frac{n-1}2}\\
iv) &\varepsilon (\alpha , \beta) \varepsilon (\beta , \alpha) &=& (-1)^{(\alpha , \beta) - m_R n_R (n-1)}\\
v) &\varepsilon (0 , \beta) &=& \varepsilon (\alpha , 0) = 1 \\
vi) &\varepsilon (-\alpha , \beta) &=& \varepsilon (\alpha , \beta)^{-1} = \varepsilon (\alpha , \beta) \\
vii) &\varepsilon (\alpha , -\beta) &=& \varepsilon (\alpha , \beta)^{-1} = \varepsilon (\alpha , \beta) \\
\end{array}\ees
\end{proposition}

\begin{proof} The first two properties follow directly from the definition. In order to prove $iii)$ we first notice that for $\alpha_i,\alpha_j\in \Delta$ , $i\ne j$, $(\alpha_i, \alpha_j) \in \{0,-1\}$. Therefore:
\be\begin{array}{rcl}
\varepsilon (\alpha , \alpha) &=& \prod_{1\le i,j\le R} \varepsilon (\alpha_i , \alpha_j)^{m_im_j}  =  \prod_{1\le i<j\le R} (-1)^{m_im_j (\alpha_i , \alpha_j)}  \prod_{1\le i \le R} (-1)^{m_i^2}\\ \\
&=& (-1)^{\sum_{1\le i<j\le R}{m_im_j (\alpha_i , \alpha_j)} + \frac12 \sum_{1\le i\le R}{m_i^2(\alpha_i , \alpha_i)} - \frac12  \sum_{1\le i\le R}{m_i^2 ((\alpha_i , \alpha_i)} - 2)} \\ \\
&=& (-1)^{\frac12 (\alpha , \alpha) - \frac12 m_R^2 ((\alpha_R , \alpha_R) - 2)} =  (-1)^{\frac12 (\alpha , \alpha) - m_R^2 \frac{n-1}2}
\end{array}\ee

Property $iv)$ follows from $iii)$ by replacing $\alpha$ with $\alpha+\beta$ and using the first two. If $\alpha = \sum m_i  \alpha_i$, $\beta = \sum n_i \alpha_i$ and $\alpha +\beta = \sum \ell_i \alpha_i = \sum (m_i+n_i) \alpha_i$ we get:

\bea{rcl}
\varepsilon (\alpha + \beta , \alpha +\beta) &=& (-1)^{\frac12 (\alpha +\beta, \alpha+\beta) - \ell_R^2 \frac{n-1}2} =
(-1)^{\frac12 (\alpha, \alpha)  +  \frac12 (\beta, \beta) +  (\alpha, \beta)- \ell_R^2 \frac{n-1}2} \\ \\
&=& (-1)^{\frac12 (\alpha , \alpha) - m_R^2 \frac{n-1}2} (-1)^{\frac12 (\beta, \beta) - n_R^2 \frac{n-1}2}
\varepsilon (\alpha, \beta)\varepsilon (\beta , \alpha )
\eea
from which the property follows. Property $v)$ is a trivial consequence of the definition, whereas properties $vi)$ and $vii)$ follow from property $v)$ together with $i)$ and $ii)$. \hfill
\end{proof}

\begin{proposition} If $\alpha, \beta ,\alpha+\beta \in \Phi$ then:
\beas{rll}
i) &\varepsilon (\alpha , \alpha) = -1 & \alpha \in \Phi\\
ii) &\varepsilon (\alpha , \beta) = - \varepsilon (\beta , \alpha) & \alpha,\beta,(\alpha+\beta)\in \Phi\qquad \text{antisymmetry}\\
iii) &\ep(\as ,\bb) = \ep(\bb, \as + \bb) & \text{if } \as, \as+\bb \in \Phi\, , \ \bb\in \Ll \\
iv) &\ep(\as ,\bb) = \ep(\bb, \as - \bb) & \text{if } \as, \as-\bb \in \Phi\, , \ \bb\in \Ll \\
\eeas
\label{remas}
\end{proposition}

\begin{proof} By \eqref{rema} if $\alpha\in \Phi_O$ then $(\alpha,\alpha) = 2$ and $m_R^2/2$ is even, hence $\varepsilon (\alpha,\alpha) = -1$. If $\alpha\in \Phi_S$ then $(\alpha,\alpha) = n+1$ and $m_R^2=1$. Therefore if $(-1)^{\frac12 (\alpha ,\alpha)-m_R^2 \frac{n-1}2}= (-1)^{\frac12 (n+1-n+1)} = -1$. As a consequence, if $\alpha, \beta ,\alpha+\beta \in \Phi$, then $-1=\varepsilon (\alpha+\beta ,\alpha+ \beta)=\varepsilon (\alpha , \beta)\varepsilon (\beta,\alpha)$, hence $\varepsilon (\alpha , \beta) = - \varepsilon (\beta , \alpha)$ \\
Finally we prove $iii)$ - similar proof for $iv)$:
\bes
\ep(\as ,\bb) = \ep(\as ,\as-\as+\bb) =\ep(\as ,\as)\ep(\as ,\as+\bb) =  - \ep(\as+\bb -\bb ,\as+\bb) = \ep(\bb, \as + \bb)
\ees
\hfill
\end{proof}

\section{T-algebras}\label{sec:ta}
{\it The essence of Mathematics lies in its freedom.}

G. Cantor\\

We now concentrate on the set of roots $T_{(a,b)}$ on a tip of the Magic Star. In order to fix one, let us consider $T_{(1,1)}$ and denote it simply by $T$. The reader will forgive us if we shall be concise whenever there is no risk of misunderstanding and use $T$ to denote both the set of roots and the set of elements in $\lms$ associated to those roots. An element of $T$ is an $\fff$-linear combination of $x_\as, x_\beta, ...$ for $\as,\bb,...$ in $T_{(1,1)}$. We give $T$ an algebraic structure with a symmetric product, thus mimicking the case $n=1$ when $T$ is a Jordan algebra.\\

Let us first show what happens in the case $n=1$, in particular for $\eo$. It is proven in \cite{Marrani-Truini-1} that $T=J_{(1,1)}$ is the Jordan Algebra $\joto$ of $3\times 3$ Hermitean matrices over the octonions $\mathbb{O}$. Let us denote by
$P_1 := E_{11}$, $P_2 := E_{22}$, $P_3 := E_{33}$ the three trace-one idempotents, whose sum is the identity in $\joto$, $E_{ii}$ representing the matrix with a 1 in the $(ii)$ position and zero elsewhere.\\
Let us identify the algebra $\dq \subset \fq\subset\es\subset\es\oplus\adc\subset\eo$ with the roots $\pm k_i \pm k_j$, $4\le i<j \le 7$ and  $P_1,P_2,P_3$ with the elements of $J_{(1,1)}$ which are left invariant by $\mathbf{d_4}$, \cite{jacob2}. This uniquely identifies $P_1 , P_2 , P_3$ with the roots $k_1+ k_8$, $k_1- k_8$ and $-k_2-k_3$. Since $\est\subset\eo$ is 3-graded, one can define a 3-linear product on $J_{(1,1)}$, \cite{mac1}, and a Jordan product through the correspondence with the quadratic formulation of Jordan algebras.\\

Still in the case of $\eo$ ($\eon$ for $n=1$), if we identify the $\eo$ generators corresponding to the roots in $\Phi_O$ with the bosons and those corresponding to the roots in $\Phi_S$ with the fermions, then the 11 bosons in $T$ are $k_1 \pm k_j$, $j=4,...,8$ and $-k_2-k_3$. The three bosons $k_1\pm k_8$ and $-k_2-k_3$, correspond to the three diagonal idempotents $\xu,\xd,\xt$, which are left invariant by $\dq$, whose roots are $\pm k_i \pm k_j \ , \ 4\le i<j \le 7$.\\
The remaining 8 vectors (bosons) and the 2x8 spinors (fermions) are linked by triality. Notice that one spinor has $\frac12(k_1-k_2-k_3 + k_8)$ fixed and even number of + signs among $\pm k_4 \pm k_5\pm k_6\pm k_7$, while the other one has $\frac12(k_1-k_2-k_3 - k_8)$ fixed and odd number of + signs among $\pm k_4\pm k_5\pm k_6\pm k_7$. So they are both 8 dimensional representations of $\dq$.\\
Therefore a generic element of $\joto$ is a linear superposition of 3 diagonal elements plus one vector and two spinors (or a bispinor). The vector can be viewed as the $(12)$ octonionic entry of the matrix (plus its octonionic conjugate in the $(21)$ position) and the two 8-dimensional spinors as the $(31)$ and $(23)$ entry (plus their respective octonionic conjugate in the $(31)$ and $(23)$ position).\\

We do the same for $n\ge 1$ and consider the most general case $\eon$ (the other two cases $\esn$ and $\estn$ being a restriction of it).\\
Let us define $\xu$, $\xd$ and $\xt$ as the elements of $\lms$ in $T$ associated to the roots $\rho_1:=k_1+ k_N$, $\rho_2:=k_1- k_N$ and $\rho_3:=-k_2-k_3$:
\be\label{notp}
\xu \lra \rho_1:=k_1+ k_N \ ; \ \xd \lra \rho_2:=k_1- k_N \ ; \ \xt \lra \rho_3:=-k_2 - k_3
\ee
They are left invariant by the Lie subalgebra $\mathbf{d_{N-4}}=\dqn$, whose roots are $\pm k_i \pm k_j \ , \ 4\le i<j \le N-1$.\\
We denote by $\tfo$ the set of roots in $T\cap \Phi_O$,  by $\tfop$ the set of roots in $\tfo$ that are not $\rho_1,\rho_2,\rho_3$ and by $\tfs$ the set of roots in $T\cap \Phi_S$. In the case we are considering, where $T=T_{(1,1)}$ we have $\tfop = \{ k_1\pm k_j\, ,\ j=4,...,N-1\}$ and $\tfs= \{ \frac{1}{2} (k_1 - k_2 - k_3 \pm k_4 \pm ... \pm k_N)\}$, even $\#$ of $+$. We further split $\tfs$ into $\tfsp = \{ \frac{1}{2} (k_1 - k_2 - k_3 \pm k_4 \pm ... + k_N)\}$ and $\tfsm= \{ \frac{1}{2} (k_1 - k_2 - k_3 \pm k_4 \pm ... - k_N)\}$. Then $v\in \tfop$ is an $8n$-dimensional vector and $s^\pm\in \tfspm$ are $2^{4n-1}$-dimensional spinors of $\dqn$.

We write a generic element $x$ of $T$ as  $x= \sum \lambda_i \xpi + \, x_v\, +\, x_{s^+}+x_{s^-}$ where
\be
x_v = \sum_{\as \in \tfop} \lambda_\as^v x_\as
\ee
\be
x_{s^\pm} = \sum_{\as \in \tfspm} \lambda_\as^{s^\pm} x_\as
\ee

We view $\lambda_\as^v$ as a coordinate of the vector $\lambda^v$ and $\lambda_\as^{s^\pm}$ as a coordinate of the spinor $\lambda^{s^\pm}$; we denote by  $\bar\lambda^v$ ($\bar\lambda^{s^\pm}$) the vectors ( spinor) in the dual space with respect to the appropriate bilinear form and view $x$ as a $3\times 3$ Hermitean matrix:

\be\label{matrix}
\left(
\begin{array}{lll}
\lambda_1 &\lambda_v &\bar  \lambda_{s^+}\\
\bar \lambda_v &\lambda_2 &\lambda_{s^-}\\
\lambda_{s^+} &\bar  \lambda_{s^-} &\lambda_3
\end{array}
\right)
\ee
 whose entries have the following $\fff$-dimensions:
\bit
\item 1 for the {\it scalar} diagonal elements $\lambda_1,\lambda_2,\lambda_3$;
\item $8n$ for the {\it vector} $\lambda_v$;
\item $2^{4n-1}$ for the {\it spinors} $\lambda_{s^\pm}$.
\eit

We see that only for $n=1$ the dimension of the vector and the spinors is the same, whereas for $n>1$ the entries in the $(12), (21)$ position have different dimension than those in the $(31),(13),(23),(32)$ position. Nevertheless we now define a symmetric product of the elements in $T$, which then becomes a generalization of the Jordan algebra $\joto$ in a very precise sense. This type of generalization of the Jordan Algebra is known in the literature, \cite{vin}, where it is called $T-algebra$.

We denote by $I$ the element $I:=\xu+\xd+\xt$ and by $I^-$ the element $I^-:=-\xub-\xdb-\xtb$ of $\bar T:= T_{(-1,-1)}$, where $\xub$, $\xdb$ and $\xtb$ are associated to the roots $-k_1- k_N$, $-k_1+ k_N$ and $k_2+k_3$.\\

We give $T$ an algebraic structure by introducing the symmetric product, see Proposition \sref{t-al}:
\be\label{jprod}
x\jp y := \frac12 [[x,I^-],y] \quad ,\quad x,y\in T
\ee

We introduce the trace (see Proposition \sref{t-al}) $tr(x)\in \fff$ for $x\in T$ in the following way:
\be\label{trace}
\text{for } x=\ell_1 \xu+\ell_2 \xd+\ell_3 \xt +\sum_{\substack{\as\ne\rho_1,\rho_2,\rho_3}}\ell_\as x_\as\quad tr(x)=\ell_1+\ell_2+\ell_3
\ee

We denote by $tr(x,y):=tr(x\jp y)$ and by $x^2:=x\jp x$. For each $x\in T$ we define
\be \label{sharp}
x^\# = x^2 -tr(x) x - \frac12 (tr(x^2) - tr(x)^2) I
\ee
and we say that $x$ is rank-1 if $x^\#=0$. Notice that $tr(x^\#)=- \frac12 (tr(x^2) - tr(x)^2)$, therefore $x^\#=0$ implies $x^2 = tr(x) x$: a rank-1 element of $T$ is either a  nilpotent or a multiple of a {\it primitive} idempotent $u\in T$: $u^2 = u$ and $tr(u)=1$.\\

Let us also introduce $N(x)\in \fff$ for $x\in T$ in the following way:
\be\label{norm}
N(x) = \frac16 \left\{ tr(x)^3 - 3 tr(x)tr(x^2) +2 tr(x^3) \right\} = \frac13 tr(x^\#,x)
\ee
where $x^2= x\jp x$ and $x^3= x\jp x^2$.\\

Since by \eqref{norm} a rank-1 element $x$ has $N(x)=0$ any element of $T$ falls into the following classification:
\bea{ll}
\text{rank-1:} & x^\# = 0\\
\text{rank-2:} & x^\# \ne 0 \ , \ N(x)=0\\
\text{rank-3:} & N(x)\ne 0\\
\eea

Examples of rank-2 and rank-3 elements are $\xu + \xd$ and $I$ respectively.\\
We finally now show the main properties of the algebra $T$.

\begin{proposition} Let $x,y\in T$ then $[[x,I^-],y]+ [[I^-,y],x]=0$.
\label{der}
\end{proposition}

\begin{proof} By the linearity it is sufficient to prove the proposition for basis elements $X_0\in \{\xub,\xdb,\xtb\}$ and basis elements $X_1, X_2$ of $T$ and the statement of the Proposition amounts to $J_0+J_2=0$, with $J_0 := [[X_0,X_1],X_2]$ and $J_2 := [[X_2,X_0],X_1]$
Set $X_0 = x_\alpha$, $\alpha\in\{-k_1\pm k_N, k_2+k_3\}$, $X_1 = x_\beta$, $X_2 = x_\gamma$. Since $X_1,X_2\in T$ then $[X_1,X_2]=0$ and $\beta+\gamma \notin \Phi\cup\{0\}$; we also have that $\as\neq\bb, \as\neq\gh$.\\

If none of the sums $\alpha+\beta$, $\alpha+\gamma$ is in $\Phi$ nor is 0, then $J_0 +J_2= 0$. Also if $\beta=
\gamma$ then $J_0 + J_2 = 0$ is trivially satisfied.\\

From now on at least one of $\alpha+\beta$, $\alpha+\gamma$, is in $\Phi\cup \{0\}$.
Suppose first that $\alpha+\beta=0$. Then $J_0 = - (\gamma, \alpha)x_\gamma$ and we have the following possibilities:
 \bit
 \item[a1)] if $\alpha + \gamma = 0$  then $\beta = \gamma$ and $J_0 +J_2= 0$ becomes trivial;
 \item[a2)] if $\alpha + \gamma \in \Phi$ then $(\gamma, \alpha)=-1$ by Proposition \sref{sproots} and $J_0 = x_\gamma$; moreover $J_2 = \varepsilon(\gamma,\alpha)\varepsilon(\alpha+\gamma,-\alpha) x_\gamma = \varepsilon(\alpha,-\alpha) x_\gamma=-x_\gamma$ (because of Propositions \sref{epsprop} and \sref{remas}) and  $J_0 +J_2= 0$ is verified;
 \item[a3)] if both $\alpha \pm \gamma \notin \Phi\cup \{ 0\}$ then $(\alpha,\gamma)=0$, by Proposition \sref{sproots} being $\alpha\in\Phi_O$, hence $J_0 =  J_2 = 0$.
 \eit
Similarly if we suppose $\alpha+\gamma=0$.\\

From now on $\alpha,\beta,\gamma\in \Phi$, $\alpha\ne\pm\beta\ne\pm\gamma\ne\pm\alpha$ and $\alpha+\beta+\gamma\in \Phi$ (notice that $\alpha+\beta+\gamma$ cannot be zero for $\as\in T_{(-1,-1)}$ and $\bb,\gh\in T_{(1,1)}$ and if $\alpha+\beta+\gamma\notin \Phi$ then obviously $J_0 = J_2 = 0$).\\

If none of $\as+\bb,\as+\gh$ is in $\Phi$ then $J_0=J_2=0$. If $\as+\bb \in \Phi$ then also $\as+\gh \in \Phi$ (and viceversa): if $\gamma\in \Phi_O$ then, by Proposition \sref{sproots}, $(\as+\bb,\gh)=-1=(\as,\gh)+(\bb,\gh)=(\as,\gh)$, implying $\as+\gh \in \Phi$; if $\bb\in \Phi_O$ and $\gamma\in \Phi_S$ then $(\bb,\gh)\in \{0,1\}$ and $(\as+\bb,\gh)=-1=(\as,\gh)+(\bb,\gh)$ implies $(\as,\gh)=-1$ and $(\bb,\gh)=0$;
if $\bb,\gamma\in \Phi_S$ then $(\as+\bb)\in \Phi_S$ and $(\as+\bb,\gh)=-n=(\as,\gh)+(\bb,\gh)$, with $(\as,\gh)\in\{ 0,\pm1\}$ since $\as\in\Phi_0$. But $(\as,\gh)=0$ implies $\bb+\gh\in \Phi$ and $(\as,\gh)=1$ implies $(\bb,\gh)=-(n+1)$ hence $\bb=-\gh$ which both contraddict the hypothesis. Hence $(\as,\gh)=-1$ and $\as+\gh\in \Phi$.\\

Finally, if $\as+\bb,\as+\gh \in \Phi$ then
\bea{l}
[[x_\as,x_\bb],x_\gh] + [[x_\gh,x_\as],x_\bb]=\\ -\eab\ega\ebg+\ega\egb\eab= \eab\ega(\egb-\ebg)\label{fjacp}
\eea
By $iv)$ of Proposition \sref{epsprop} plus Proposition \sref{sproots} and \eqref{rema} we have $\ebg \egb = (-1)^{(\bb,\gh)-m_Nn_N(n-1)}$. If $\bb\in\Phi_O$ or $\gh\in \Phi_O$ then  $\ebg \egb = (-1)^{(\bb,\gh)}$ and $(\bb,\gh)\in\{0,1\}$. Suppose $(\bb,\gh)=1$. Then either $\as+\gh\in\Phi_O$ or $\bb\in\Phi_O$ and $(\as+\gh,\bb)=-1$; but also $(\as+\gh,\bb)=(\as,\bb)+(\gh,\bb)=-1+1=0$, a contraddiction. If both $\bb,\gh\in\Phi_S$  then $(\as+\gh,\bb)=-n$ and $(\as,\bb)=-1$. So $(\bb,\gh)=-(n-1)$ and $\ebg \egb = (-1)^{(\bb,\gh)-m_Nn_N(n-1)}= (-1)^{-(n-1)\pm(n-1)}=1$, therefore $\ebg=\egb$ and \eqref{fjacp} is zero.\\
This ends the proof of Proposition \sref{der}. \hfill
\end{proof}

\begin{proposition}\label{t-al}
The product $x,y\to x\jp y$ is symmetric. With respect to this product the element $I\in T$ is the identity  and the elements $\xu , \xd , \xt\in T$ are trace-one idempotents, hence are rank-1 elements of $T$. All $x_\as \in T$ but  $\xu , \xd , \xt$ are nilpotent and trace-0, hence they are also rank-1 elements of $T$. The form $x \to tr(x)$ is a trace form, namely $tr(x,y)$ is bilinear and symmetric in $x,y$ and the associative property $tr(x\jp y,z)=tr(x,y\jp z)$ holds for every $x,y,z \in T$.
\end{proposition}

\begin{proof} By Proposition \sref{der}:
\be [[x,I^-],y] + [[I^-,y],x] = 0\ \Rightarrow\ [[x,I^-],y] = [[y,I^-],x]
\ \Rightarrow\ x\jp y = y\jp x
\ee
We also have
\be
[\xpi,-\xpjb]=\delta_{ij} \hri \ , \ i,j=1,2,3
\ee
and for any $\as \in T_{(r,s)}$, being $\rho_1+\rho_2+\rho_3 = 2k_1-k_2-k_3$:
\be
I\jp x_\as = \frac12 [[I,I^-],x_\as] = \frac12 [\hru+\hrd+\hrt,x_\as] = \frac12 (\rho_1+\rho_2+\rho_3,\as) x_\as = x_\as
\ee
By linearity this extends to any $x\in T$.\\
Moreover:
\be \xpi\jp \xpi=\frac12 [[\xpi,I^-],\xpi] = \frac12 [h_{\rho_i},\xpi] = \frac12 (\rho_i,\rho_i) \xpi = \xpi
\ee
Obviously $tr(\xu)=tr(\xd)=tr(\xt)=1$.\\
We now show that $x_\as \in T$ is nilpotent if $\as\ne k_1\pm k_N, -k_2-k_3$.\\
Suppose $x_\as\jp x_\as\ne 0$ then either $2 \as - k_1 - k_N$ or $2 \as - k_1 + k_N$ or $2 \as + k_2 + k_3$ is a root. If $\as\in \Phi_O$ then $2 \as - k_1 \mp k_N = k_1 \pm 2 k_i \mp k_N$ is a root if and only if $k_i=k_N$, whereas $2\as + k_2 +k_3$ is not a root. If $\as\in \Phi_S$ then nor $2 \as - k_1 \mp k_N$ nor $2\as + k_2 +k_3$ are roots. Therefore $\as\ne k_1\pm k_N, -k_2-k_3$ implies $x_\as^2=0$ and, obviously, $tr(x_\as)=0$ and $x_\as$ be rank-1.\\
Finally, from the definition of the product $x,y\to x\jp y$ and the definition of the trace it easily follows that $tr(x,y)$ is bilinear. The associativity of the trace is proven by direct calculation.
\end{proof}

\section{Future Developments}\label{sec:fd}
{\it There are more things in heaven and earth, Horatio, than are dreamt of in your philosophy.}

``Hamlet'' 1.5.167-8, W. Shakespeare\\

We have presented a periodic infinite chain of finite generalisations of the exceptional structures, including $\eo$, the exceptional Jordan algebra $\mathbf{J}_{\mathbf{3}}^{\mathbb{O}}$ (and pair), and the octonions.
We have demonstrated that the exceptional Jordan algebra $\mathbf{J}_{\mathbf{3}}^{\mathbb{O}}$ is part of an infinite family of finite-dimensional matrix algebras (corresponding to a particular class of Vinberg's cubic T-algebras \cite{vin}). Correspondingly, we have proved that $\eo$ is part of an infinite family of MS algebras, that resemble lattice vertex algebras.

\bigskip

We are currently working on several topics concerning the mathematical structure presented in this paper. Many of them are at an advanced stage, some are completed. In particular we will expand in future papers on the algebra of derivations of the algebras $\lms$ and $T$, the lattice, the modular forms and the theta functions related to EP.\\
There are also several topics that we are planning to develop in the future, strictly related to Quantum Gravity. In particular, a model for interactions based on EP which includes gravity and the expansion of space-time, starting from a singularity. We aim at a new perspective of elementary particle physics at the early stages of the Universe based on the idea that interactions, defined in a purely algebraic way, are the fundamental objects of the theory,
whereas space-time, hence gravity, are derived structures.


An interesting venue of research stemming from EP is to study the higher dimensional
weight vectors of algebras akin to lattice vertex algebras (the original
motivation for Borcherds' definition of vertex algebras \cite{29}), that project to a
``Magic Star'' structure. Such extended ``Magic Star'' algebras make use of an
\textit{asymmetry function} \cite{kac, graaf}, that acts like the cocycle of a lattice vertex
algebra which gives a twisted group ring $\mathbb{C}_{\epsilon }[\Lambda ]$
over an even lattice $\Lambda $. This gives algebras that span higher
dimensional lattices, beyond that of the self-dual $D=8$ lattice of $\eo$,
and allows one to potentially probe the symmetries of the heterotic string and moonshine
\cite{27,28}, as well as to formulate a $T$-algebra generalization of noncommutative geometry.
Specifically, EP gives a novel algebraic method to study even self-dual
lattices, such as the $\eo\oplus \eo$ and Leech lattices, which already
have a well known connection to the Monster vertex algebra and $D=24$
bosonic string compactifications \cite{30, 31, 32}.\\

With an infinite family of new algebras that extend the exceptional Lie
algebras, EP gives a fresh new toolkit for studying emergent spacetime
and quantum gravity, in dimensions beyond those previously explored, using
spectral techniques applied to an infinite class of cubic $T$%
-algebras.

An interaction based approach might also have a Yang-Mills interpretation, as EP hints at higher-dimensional (possibly super) Yang-Mills theories (see \textit{e.g.} \cite{sezgin}), that may exhibit some kind of triality.\\

As the exceptional Jordan algebra $\mathbf{J}_{\mathbf{3}}^{\mathbb{O}}$ is to $\eo$, so is the particular infinite class of cubic T-algebras to the higher EP algebras, when projected along an $\mathbf{a}_{2}$.  Therefore, physically, the class of cubic T-algebras under consideration allows one to consistently generalize exceptional quantum mechanics of Jordan, Wigner and von Neumann \cite{JWVN}, because they go beyond the formally real Jordan algebraic classification of quantum-mechanical self-adjoint operators; in a sense, the classification of quantum-mechanical observables can only probe the T-algebras up to the exceptional Jordan algebra.\\

It is the authors' hope that a synthesis of spectral algebraic geometry and
EP can unify many approaches to unification and quantum gravity, and provide
a new lens for searching a non-perturbative theory of all matter, forces and spacetime.

\section*{Acknowledgments}

The work of PT is supported in part by the \textit{Istituto Nazionale di Fisica Nucleare} grant In. Spec. GE 41.


\bibliographystyle{amsplain}

\end{document}